\documentclass[%
 reprint, %linenumbers,
superscriptaddress,
nobibnotes,
 amsmath,amssymb,
 aps, 
 prl,
floatfix,
showkeys
]{revtex4-2}
\usepackage{ amsthm}
\usepackage{graphicx}% Include figure files
\usepackage{dcolumn}% Align table columns on decimal point
\usepackage{bm}% bold math
\usepackage{enumitem}
\usepackage{color}
 \usepackage{extpfeil}
\usepackage{titlesec}
\usepackage{multirow}
\usepackage{makecell}
\usepackage{threeparttable}
\usepackage{array}
\newcolumntype{C}[1]{>{\centering}p{#1}}
\newtheorem{theorem}{\textbf{Theorem}}
\newtheorem{proposition}{\textbf{Proposition}}
\newtheorem{lemma}{\textbf{Lemma}}
 
\renewcommand{\eqref}[1]{Eq. (\ref{#1})}

\begin{document}

\preprint{AIP/123-QED}

\title{Exact first passage time distribution for  nonlinear chemical reaction networks \\ II: monomolecular reactions and a $A + B \rightarrow C$ type of second-order reaction with arbitrary initial conditions}

\author{Changqian Rao}
 \affiliation{School of Mathematical Sciences, Fudan University, Shanghai, 200433, China}
 \affiliation{ 
\mbox{Research Institute of Intelligent Complex Systems, Fudan University, Shanghai, 200433, China}
}

\author{David Waxman}
 \affiliation{ 
\mbox{Institute of Science and Technology for Brain-Inspired Intelligence, Fudan University, Shanghai, 200433, China}
}
 \affiliation{ 
Key Laboratory of Computational Neuroscience and Brain-Inspired Intelligence (Fudan University), Ministry of Education, China
}
 \affiliation{ 
\mbox{MOE Frontiers Center for Brain Science, Fudan University, Shanghai 200433, China}
}
 \affiliation{ 
Zhangjiang Fudan International Innovation Center, Shanghai, China
}

 \author{Wei Lin
}
 \affiliation{School of Mathematical Sciences, Fudan University, Shanghai, 200433, China}
 \affiliation{ 
\mbox{Research Institute of Intelligent Complex Systems, Fudan University, Shanghai, 200433, China}
}
 \affiliation{ 
\mbox{Institute of Science and Technology for Brain-Inspired Intelligence, Fudan University, Shanghai, 200433, China}
}
 \author{Zhuoyi Song%\textsuperscript{*} %$^{3,4,5,6}$%\thanks{Corresponding author: songzhuoyi@fudan.edu.cn}
}
 \email[E-mail me at: ]{songzhuoyi@fudan.edu.cn}
\affiliation{ 
Key Laboratory of Computational Neuroscience and Brain-Inspired Intelligence (Fudan University), Ministry of Education, China
}
 \affiliation{ 
\mbox{MOE Frontiers Center for Brain Science, Fudan University, Shanghai 200433, China}
}
 \affiliation{ 
Zhangjiang Fudan International Innovation Center, Shanghai, China
}
% \affiliation{%
%  Authors' institution and/or address\\
%  This line break forced with \textbackslash\textbackslash
% }%

\date{\today}% It is always \today, today,

\begin{abstract}

In biochemical reaction networks, the first passage time (FPT) of a reaction quantifies the time it takes for the reaction to first occur, from the initial state. While the mean FPT historically served as a summary metric, a far more comprehensive characterization of the dynamics of the network is contained within the complete FPT distribution. The relatively uncommon theoretical treatments of the FPT distribution that have been given in the past have been  confined to linear systems, with zero and first-order processes. Recently, we presented theoretically exact solutions for the FPT distribution, within nonlinear systems involving two-particle collisions, such as $A + B \rightarrow C$. Although this research yielded invaluable results, it was based upon the assumption of initial conditions in the form of a Poisson distribution.
This somewhat restricts its relevance to real-world biochemical systems, which frequently display intricate behaviour and initial conditions that are non-Poisson in nature. Our current study extends prior analyses to accommodate \textit{arbitrary initial conditions}, thereby expanding the applicability of our theoretical framework and providing a more adaptable tool for capturing the dynamics of biochemical reaction networks.

\end{abstract}


\maketitle

%\tableofcontents
\section{Introduction}

In biochemical networks, the first passage time (FPT) refers to the time it takes for an event to take place or for a state to reach a specific threshold for the first time. The FPT is often defined in distinct biological contexts, such as the time it takes for gene expression after a stimulus \cite{ghusinga2017first}, the time required for a signaling pathway to produce a cellular response in signal transduction \cite{huang2021relating}, or the time until a substrate is converted into a product, in enzyme kinetics \cite{polizzi2016mean}. 

The concept of FPT is particularly important when the biochemical networks are confined to a small subcellular domain, where the number of molecules and proteins is low. In these cases, reactions and interactions among molecules are stochastic rather than deterministic in character. As a result, discreteness and stochasticity are inherent properties. Indeed, stochastic processes are widely used as a modeling methodologies to describe the dynamics of biochemical networks \cite{gillespie2013perspective}. 

Analyzing the full FPT distribution of a stochastic process, rather than the mean or global FPT, can expose
much more information about the timing, efficiency, and reliability of the underlying biochemical reaction system \cite{ham2024stochastic,polizzi2016mean}. The full FPT distribution provides detailed insights into the underlying chemical kinetics and the regulatory mechanisms that operate \cite{huang2021relating}. 

Obtaining the full FPT distribution of a biochemical network typically involves solving the underlying \textit{chemical master equation} (CME), when subject to appropriate boundary conditions \cite{iyer2016first}. 
The CME describes the time evolution of the probabilities of each possible state of the system. Theoretical treatments of the CME are rare \cite{gillespie2013perspective}, and time-dependent CME solutions are typically estimated by simulations, based on kinetic Monte Carlo methods \cite{cao2006efficient}, such as the Gillespie algorithm \cite{gillespie2013perspective}, or approximated with moment closure schemes \cite{smadbeck2013closure} or or mean-field approximations \cite{cao_linear_2018}. 

Both simulated and approximate results show that initial conditions significantly influence the shape of the FPT distribution \cite{woods2024analysis,nyberg2016simple}. The initial conditions affect the system's proximity 
to the target state. Thus when a system starts further away from the target (as seen in chemical reactions, diffusion processes, or population dynamics), the mean FPT tends to be longer, reflecting the greater time required to reach the desired state \cite{woods2024analysis}. Importantly, it is not just mean values associated with
the initial state that matter, but also the distribution itself. 

The distribution of initial states determines the starting energy or state configuration of the system, influencing how readily the system can transit to other states and resulting in different FPTs \cite{nyberg2016simple}. More intuitively, if the initial distribution is broad, the system may be able to explore a wider range of paths, leading to distinct phases or system behaviors, resulting in multimodal or intrinsically complex FPT distributions \cite{woods2024analysis}. 

These complexities indicate how a system's initial configuration intricately shapes the subsequent dynamics and characteristics of the FPT. This raises a question of the underlying theory: how do initial conditions influence 
the FPT distribution? 

To the best of our knowledge, this question has only been addressed in linear systems where general time-dependent exact solutions of the CME are available. Such systems include biochemical networks with monomolecular reactions \cite{jahnke2007solving} or, more recently, other first-order reactions \cite{vastola2021solving}. 
%There, it says that 

However, widely presented bimolecular reactions, which serve as fundamental building blocks in biochemical reaction networks, are governed by highly nonlinear models, making exact solutions to the CME difficult to obtain \cite{gillespie2013perspective}. As a result, there are no general conclusions on how initial conditions affect the FPT distribution in nonlinear biochemical reaction networks.

In the limited body of theoretical work that provides exact FPT distributions for specific classes of biochemical networks with second-order reactions \cite{anderson_time-dependent_2020}, initial conditions are often subject to restrictions. Anderson et al. developed a general CME solution for networks with second-order reaction complexes, providing: (i) the system began with an initial distribution of  Poisson-product form, and (ii) which then retained
this form over its time evolution.

Previously, we derived an exact FPT distribution for biochemical networks involving a single $A + A \rightarrow C$ reaction, characterized by collisions between identical particles, providing the numbers of all molecular types were initially all zero \cite{rao2023analytical}. More recently, we obtained an exact FPT distribution for a biochemical network involving an $A + B \rightarrow C$ type of second-order reaction occurring downstream of a series of monomolecular reactions \cite{rao2024exact}. While this allowed for analysis of systems with interactions between different particles, the initial conditions were still limited to distributions of Poisson-product form \cite{rao2024exact}.

The present paper seeks to broaden our exact solution/analysis of the FPT distribution, based on initial
conditions with a Poisson-product form \cite{rao2024exact}, and generalizing it to arbitrary initial conditions. 
We note that it is essential to focus solely on discrete distributions, as molecular counts in stochastic biochemical systems are inherently discrete.

In particular, we demonstrate that our previous theorem can be extended to cover arbitrary initial discrete distributions. 
This follows from two facts: (i) a single delta distribution can be represented as a weighted sum of Poisson distributions, (ii) an arbitrary (discrete) distribution can be represented as a weighted sum over single delta distributions. As a consequence of (i) and (ii), an arbitrary discrete distribution can be expressed as a weighted sum of  distributions of Poisson product form (see Appendix \ref{Appendix:B} for further details). 
This approach thus  allows our previous theorem to be extended to arbitrary initial discrete distributions
that are distinct from those of Poisson product form.

Our findings demonstrate that changes in initial conditions can lead to significant variation in 
the FPT distribution. {\color{black} Additionally, we establish that the FPT distribution exhibits linearity with respect to initial conditions for the systems studied, indicating that shifts in the initial condition result in proportionate
shifts in the FPT distributions}. Collectively, these insights emphasize the importance of conducting theoretical analyses of FPT distributions under arbitrary initial conditions.

\section{Problem formulation}

We examine a well-mixed biochemical reaction system consisting of N+1 species, denoted by the column 
vector $\bm{S}^*=[S_0, S_1,S_2,\dots S_N]^{\top}$, undergoing $M+1$ reactions $R_m$ where $m=0,1,2,3,....,M$. For $m\geq 1$, each reaction ($R_m$) is monomolecular, and  does not involve the species $S_{0}$. These reactions can be categorized into one of the following three types: 
\begin{equation}
\begin{split}\label{Cas0}
\begin{split}
   S_{j}&\overset{a_{m}(t)}{\rightarrow} S_{k} \mbox{ conversion }(j\neq k,j\neq 0,k\neq 0); \\
*&\overset{a_{m}(t)}{\rightarrow} S_{k} \mbox{ production from source or inflow }(k\neq 0); \\
S_{j}&\overset{a_{m}(t)}{\rightarrow} * \mbox{  degradation or outflow }(j\neq 0).   
\end{split}
\end{split}
\end{equation}
where $a_{m}(t)$ for $m=0,1,\dots,M$ denote the reaction rate `constants' of the $m$th reaction. The $a_{m}(t)$ are assumed to be nonnegative and they can, potentially be time-dependent, and hence vary with $t$. 

The reaction $R_{0}$ is a second-order reaction of the form 
\begin{equation}
\begin{split}\label{Cas1}
% \bm{y}_m\cdot \bm{S}\overset{a_{m}(t)}{\rightarrow}& \bm{y'}_m\cdot \bm{S},\\
S_1+S_2\overset{a_{0}(t)}{\rightarrow}&S_0
\end{split}
\end{equation}
where $S_1$ and $S_2$ are reactant species that interact to produce $S_0$ with a time-dependent reaction rate $a_0(t)$. This type of reaction involves the collision of two different species and is non-linear in nature, distinguishing it from the monomolecular reactions ($R_m$ for $m\geq 1$).

This is a nonlinear biochemical network that is composed of a second-order reaction of the type $A + B \rightarrow C$, following a series of monomolecular reactions. Unlike general first-order reactions, this system excludes (auto)catalytic reactions of the form $S_j \rightarrow S_j + S_k$, nor splitting reactions $S_j \rightarrow S_l + S_k$ (where $j \neq l,k$). 

We can rewrite the system of reactions in the more compact way: 

\begin{equation}
\begin{split}\label{apb_reac}
\bm{Y}_m\cdot \bm{S}^*\overset{a_{m}(t)}{\rightarrow}& \bm{Y'}_m\cdot \bm{S}^*,\\
S_1+S_2\overset{a_{0}(t)}{\rightarrow}&S_0.
\end{split}
\end{equation}
where $\bm{Y}_m$ and $\bm{Y}_m'$ are the reactant/product coefficient of reaction $R_m$. For each monomolecular reaction, the sum of the absolute values of all components of $\bm{Y}_m$, and $\bm{Y}_m'$ must not exceed one, which we express as $\|\bm{Y}_m\|_1\leq 1$ and $\| \bm{Y}'_m\|_1\leq 1$, respectively.

To derive the exact FPT distribution for the second-order reaction, we introduce an auxiliary species $S_0$ to track the timing of the second-order reaction \cite{iyer2016first}. Let the count of  $S_0$ molecules be denoted by $x_0 \equiv x_0(t)$, where  $x_0$ is a non-decreasing function of time $t$. 

If the second-order reaction has not occurred by time $t$, then 
$x_0(t) = 0$, indicating that the FPT is greater than $t$. Thus, the complementary cumulative distribution function (CCDF) of the FPT is given by:
\begin{equation}
\begin{split}\label{fpt_aux}
{\rm P}(FPT>t) = P(x_0(t) = 0). \\
\end{split}
\end{equation}

Let $\bm{X}^*(t)$ represent the state vector that counts the number of molecules of $\bm{S}^*$ at time $t$. In our previous work, we derived an exact theoretical FPT distribution for system \eqref{apb_reac} under the condition that the initial state, $\bm{X}^* (0)$, follows a Poisson product form:

\begin{equation}
\begin{split}\label{sde_ini}
{\rm P}(\bm{X}^*(0)=\bm{x}^*)&=\prod_{i=0}^N \frac{ \theta_{i}^{x_i}}{x_i !} \exp(-\theta_{i}),
\end{split}
\end{equation}
where $\bm{x}^*=[x_0,x_1,\dots,x_N]^{\top}$ denotes a specific system state, $l_{0}=0$ and $\bm{\theta}=[\theta_{1},\dots,\theta_{N}]^\top$ is a vector of constants. 

In the present paper, we address the problem of extending the exact FPT distribution to accommodate \textit{arbitrary initial conditions}. In the subsequent sections, we first provide background on our earlier work, which offers the exact solution to \eqref{general_cme_element} for systems with initial conditions that follow a Poisson product form. Then, in the results section, we present the revised formulas for delta initial conditions and for any arbitrary initial conditions.

\section{Background}

Lemma 1 presents our earlier theorem for the exact solution of \eqref{general_cme_element} when the system's initial conditions conform to the Poisson-product form. In this paper, we use the following variables and notation:

\begin{enumerate}
     \item %$\mathbf{X}=\left[\begin{array}{c}
    %      x_{0}  \\
    %      \mathbf{x} 
    % \end{array}\right]$
    $\bm{X}^*= [X_0,\bm{X}^\top]^\top$
    denotes the complete state vector of the system, representing the numbers of different species of $\bm{S}^*$.
    % $\bm{S}^*=\left[\begin{array}{c}
    %      S_0  \\
    %      \bm{S} 
    % \end{array}\right]$. 
    The quantity $\bm{X}=[X_{1},\cdots,X_{N}]^{\top}$ denotes the numbers of different species of $\bm{S}=[S_1,S_2,\dots,S_N]^\top$.

\item $\bm{Y}_m=\left[
         y_{m,0},
         \bm{y}_m 
    \right]$ and $\bm{Y'}_m=\left[
         y'_{m,0},  
         \bm{y'}_m 
    \right]$ are stoichiometric coefficient vectors of the $m$'th reaction in system \eqref{apb_reac}, and $\bm{y}_m=[y_{m,1},y_{m,2},\dots,y_{m,N}]$, $\bm{y'}_m=[y'_{m,1},y'_{m,2},\dots,y'_{m,N}]$ denote parts of stoichiometric coefficient vectors, excluding $y_{m,0}$ and $y'_{m,0}$ for the auxiliary species $S_0$.

\item $a_{0}(t)$ denotes the rate `constant' at time $t$ of the second-order reaction,
while $a_{m}(t)$ ($m=1,\dots,M$) denotes the rate `constant' at time $t$ of the $m$'th reaction
 (which is either a zero or first-order reaction).

\end{enumerate}

The CMEs for the system given in \eqref{apb_reac} correspond to a differential equation for the probability distribution of $\bm{x}^*$ at time $t$ \cite{gillespie2013perspective}: 
\begin{small}
\begin{equation}
\label{general_cme_element}
\displaystyle\frac{{\rm d} P(\bm{x}^*, t )}{{\rm d} t} \! =\!\sum\limits_{k=0}^{M}\left[ P\left(\bm{x}^*-\bm{v}_k, t \right) c_k\left(\bm{x}^*-\bm{v}_k,t\right)\!-\! P(\bm{x}^*, t ) c_k(\bm{x}^*,t)\right],
\end{equation}
\end{small}%\end{equation} 
where $\bm{v}_k$ is the transition vector for the $k$'th reaction, and $c_k(\bm{x}^*,t)$ is the propensity function for the $k$'th reaction, i.e., the probability that the $k$'th reaction occurs in state $\bm{x}^*$, at time $t$. The 
quantity $c_k(\bm{x}^*,t)$ equals $a_k(t)(\bm{x}^*)^{\bm{Y}_k}$, where, for any vectors with $d$ components, a vector to the power of a
vector is defined by $\mathbf{u}^{\mathbf{v}}\overset{\rm def}{=}\prod^d_{i=1}u_i^{v_i}$ and we adopt the convention that $0^0=1$. Therefore, for $k=0$, the reaction is given in \eqref{apb_reac}, thus, $c_0(\bm{x}^*,t)=a_0(t)x_1x_2$. Because the first $M$ reactions ($k=1,\dots, M$) are monomolecular reactions, $c_k(\bm{x}^*,t)$, $k=1,\dots, M$ are linear functions of states.

Lemma 1 was given in terms of the following variables/notations: 

\begin{enumerate}[label=(\roman*)]

\item $\boldsymbol{\lambda}(t)=[\lambda_{1}(t),\dots
,\lambda_{N}(t)]^\top$ is a column vector containing the mean of all numbers of the species in $\bm{S}$ at
any time $t$ ($t\geq0$), and $\boldsymbol{\Lambda}(t)= \left[
     \lambda_{S}(t)  ,
     \boldsymbol{\lambda}(t) 
  \right]$ is a column vector containing the mean of all numbers of the species in $\bm{S^*}$ at
any time $t$ ($t\geq0$).

\item %$\boldsymbol{\Lambda}(0)=[\lambda_{0}(0),\lambda_{1}(0),\dots
%,\lambda_{N}(0)]^\top$ 
$\boldsymbol{\Lambda}(0)= \left[
     \lambda_{S}(0)  ,
     \boldsymbol{\lambda}(0) 
 \right]^\top$denotes the initial mean number of each species. 

\item $\mathcal{M}_{1}$ denotes an $N\times N$ matrix with components
$\eta_{i,j}^{1}$, and $\mathcal{M}_{2}$ denotes an $N\times1$ vector with
components $\eta_{i}^{2}$, where the components are given by:
\begin{equation}%
%\[
\eta_{i,j}^{1}=\sum_{m:\bm{y}_{m}=e_{i}}a_{m}(t)(y_{m,j}^{\prime}%
-y_{m,j}),\quad\eta_{i}^{2}=\sum_{m:\bm{y}_{m}=\mathbf{0}}a_{m}%
(t)(y_{m,i}^{\prime}-y_{m,i}),
%\]
\end{equation}
in which $e_i$ is a vector where only element $i$ is $1$, and all other elements
are $0$, while $y_{m,i}^{\prime}$ and $y_{m,i}$ are the $i$'th component
of $\bm{y}_{m}^{\prime}$ and $\bm{y}_{m}$, respectively,

\item $\mathcal{N}_{1}$ and $\mathcal{N}_{2}$ denote $N\times N$ matrices with all elements zero, except the upper left $2\times2$ block, with
\begin{equation}\label{n1n2_def}
\mathcal{N}_{1}=\left(
\begin{array}
[c]{ccc}%
a_{S} & 0 & \cdots\\
0 & -a_{S} & \cdots\\
\vdots & \vdots & \ddots
\end{array}
\right)  ,\quad\mathcal{N}_{2}=\left(
\begin{array}
[c]{ccc}%
\mathrm{i}a_{S} & 0 & \cdots\\
0 & \mathrm{i}a_{S} & \cdots\\
\vdots & \vdots & \ddots
\end{array}
\right),
\end{equation}
where %$a_S=\frac{\sqrt{2a_0}}{2}$  
$a_S=\sqrt{2a_0}/2$. 

\item $\lambda_{S}(t)$ and $\boldsymbol{\lambda}(t)=[\lambda_{1}(t),\dots
,\lambda_{N}(t)]^\top$ are a stochastic variable and a column stochastic vector, respectively, obey the stochastic 
differential equations (SDEs): 
\begin{equation}
\begin{split}\label{sde_formu1}
{\rm d}\boldsymbol{\lambda}&=\left(\mathcal{M}_1\boldsymbol{\lambda}+\mathcal{M}_2\right){\rm d}t + \mathcal{N}_1\boldsymbol{\lambda}{\rm d}W^1_t+\mathcal{N}_2\boldsymbol{\lambda}{\rm d}W^2_t,\\
{\rm d}\lambda_{S}&= \left( a_S \lambda_1-a_S \lambda_2 \right){\rm d}W^1_t + \left( {\rm i}a_S \lambda_1+{\rm i}a_S \lambda_2 \right){\rm d}W^2_t,
\end{split}
\end{equation}
subject to initial condition $ \boldsymbol{\lambda}(0)=\bm{\theta}$, $ \lambda_{S}(0)=0$

\end{enumerate}

\vspace{0.5em}

\begin{lemma}
\label{theo1}

For the system in \eqref{apb_reac}, the complementary cumulative probability distribution of the FPT
%associated with \eqref{general_cme_element} 
equals the probability of occurrence of states with $x_0=0$:
\begin{equation}\label{eq_coro1}
\begin{split}
P(FPT>t)&= \sum_{\bm{x}^*|x_0=0}P(\bm{x}^*,t)
=\left<\exp(\lambda_{S})\right>,
\end{split}
\end{equation}
where $<...>$ denotes an expectation operation.

\end{lemma}

\vspace{0.5em}

\section{Results}

In Section A, we start by examining the case of delta initial conditions, which correspond to a system initialized in specific states. Section B gives details of a numerical method for calculating the theoretical FPT formula. Section C extends the analysis to systems with arbitrary non-Poisson discrete initial distributions, treating these as combinations of discrete delta distributions. In Section D, we verify our theoretical predictions by comparison with simulated FPT distributions that were  obtained via the Gillespie algorithm. 

\subsection{Exact representation for the auxiliary chemical master equation with delta initial conditions}

We begin by proving that any non-Poisson discrete distribution can be expressed as a weighted sum of different Poisson distributions. Consequently, an $N$-dimensional delta distribution of the state $\bm{x}$ can be represented as a weighted sum of distributions in Poisson-product form. Using this, we extend our previous theoretical analysis of FPT distributions from Poisson-product form initial conditions to arbitrary initial conditions. 

We define an $N$-dimensional delta distribution, $\delta_{\bm{Z}_0}$, over the $N$-dimensional set of natural numbers $\mathbb{N}^{N}$ as follows For any arbitrary $\bm{x} = [x_1, x_2, \dots, x_N]^\top \in \mathbb{N}^{N}$, $\delta_{\bm{Z}_0}(\bm{x}) = 1$ if and only if $\bm{x} = \bm{Z}_0$, otherwise, $\delta_{\bm{Z}_0}(\bm{x}) = 0$.

\begin{proposition}\label{prop:delta_poisson}
Any $N$-dimensional delta distribution, $\delta_{\bm{Z}_0}$, can be represented as a weighted sum of distributions of the Poisson-product form (refer to Appendix \ref{Appendix:B} for proof):
\begin{equation}\label{delta_poisson}
\begin{split}
\delta_{\bm{Z}_0}(\bm{x})
&=\sum_{k_1,k_2,\dots,k_N =0}^{\infty} c_{\bm{k}}\prod_{j=1}^N \frac{(\theta^{(k_j)}_j)^{x_j}}{x_j!}\exp(-\theta^{(k_j)}_j ),
\end{split}
\end{equation}
where $\bm{k}=[k_1,\dots,k_N]^\top$ is the index vector; $c_{\bm{k}}$ is the weight coefficient; and $\theta^{(k_j)}_j$ is  {\color{black}the parameter of the $k\operatorname{th}$} distribution component within the Poisson-product-form.
\end{proposition}

In reality, there is a linear relationship between the FPT distribution and the initial distributions of the states. This can be attributed to two key factors: 1) Equation \eqref{eq_cme} shows that the FPT distribution is a linear combination of the distributions of all states when the non-decreasing state, $x_0$, is set to zero.; and 2) The CME presented in \eqref{general_cme_element} is a homogeneous linear ordinary differential equation (ODE):

% Then the FPT distribution is represented as a linear sum of the CME's solution:
\begin{equation}\label{eq_cme}
\begin{split}
P(FPT>t)&= \sum_{\bm{x}^*|x_0=0}P(\bm{x}^*,t).
\end{split}
\end{equation}

As a consequence of this linear mapping between FPT distributions and the system's initial conditions, the FPT distribution for arbitrary initial conditions can be expressed as a sum of the FPT distributions for the system in \eqref{apb_reac}, where the initial conditions follow the Poisson-product form. This leads us to the result in Theorem 1.

\begin{theorem}\label{theo2}
For system in \eqref{apb_reac}, when the initial state is a fixed value $\textbf{Z}_0=[z_{1},z_{2},\dots,z_{N}]$, i.e., the initial distribution is $\delta_{Z_0}$, the FPT distribution of the second-order reaction is: 
\begin{equation}
\begin{split}\label{eq_coro2}
P(FPT>t)&=\left(\nabla_{\boldsymbol{\theta}} +1\right)^{\textbf{Z}_0}\left<\exp(\lambda_S(t;\boldsymbol{\theta}))\right>|_{\boldsymbol{\theta}=\textbf{0}}.
\end{split}
\end{equation}
where $\nabla_{\boldsymbol{\theta}}$ is the gradient operator with respect to a vector of variables $\boldsymbol{\theta}=[\theta_1,\dots,\theta_N]$, and $\lambda_S(t;\boldsymbol{\theta} )$ is the solution of \eqref{sde_formu1}, given $\lambda_S(0)=0$ and
$\boldsymbol{\lambda}(0)=\boldsymbol{\theta}$.
\end{theorem}

Proof for theorem \ref{theo2} is presented in Appendix \ref{Appendix:A}.

The main difference of Theorem \ref{theo2} and our previous result shown in Lemma \ref{theo1} is the differential operator $\left(\nabla_{\boldsymbol{\theta}} +1\right)^{\textbf{Z}_0}|_{\boldsymbol{\theta}=\textbf{0}}$.
%{\color{blue} a shorter proof uses $f(d/d\theta+1)=e^{-\theta}f(d/d\theta)e^{\theta}$}
This operator transforms a Poisson product distribution into a delta distribution. 

To illustrate the transformation, consider a one-dimensional Poisson distribution parameterized by $\theta$:
\begin{equation}\label{onedim_ini}
p(x;\theta)=\left\{\begin{array}{cc}
    \frac{ \theta^{x}}{x !} \exp(-\theta). & x\geq 0  \\
     0 & x<0. 
\end{array} 
\right.
\end{equation}
The associated one-dimensional differential operator is defined as:
\begin{equation}
\begin{split}\label{onedim_dif}
\left.\left(\frac{{\rm d}}{{\rm d}\theta}+1\right)^{z}\right|_{\theta=0}
\end{split}
\end{equation} 
where $z$ is an integer. This operator ensures that the Poisson distribution \eqref{onedim_ini}  is transformed into the delta distribution $\delta_{z}$.  

Through a straightforward calculation, applying $\left(\frac{{\rm d}}{{\rm d}\theta}+1\right)^{z}$ modifies $p(x;\theta)$  to a shifted distribution 
$p_1(x;\theta)=p(x-z;\theta)$. Evaluating this result at $\theta = 0$ leads to $p_1(x;0)=\delta_{z}(x)$. Thus, the operator $\left(\nabla_{\boldsymbol{\theta}} +1\right)^{\textbf{Z}_0}|_{\boldsymbol{\theta}=\textbf{0}}$ effectively maps the Poisson distribution to the corresponding delta distribution. By analogy, the higher-dimensional operator $\left(\nabla_{\boldsymbol{\theta}} +1\right)^{\textbf{Z}_0}|_{\boldsymbol{\theta}=\textbf{0}}$, extends this transformation to a Poisson-product distribution, converting it into a delta distribution in multidimensional space.
% To explain how the differential operator can transform a Poisson product 
% initial distribution to a delta one, we start with \eqref{eq_coro2} to shortly show the transformation process. Because the $\theta$ only govern the initial distribution of the system state and $\left<\exp(\lambda_S(t;\boldsymbol{\theta}))\right>$ is linear with the $\theta$, the differential of $\left<\exp(\lambda_S(t;\boldsymbol{\theta}))\right>$ with respect to $\theta$ only change the initial system state.  

\vspace{0.1em}

\subsection{Numerical approximations for the FPT distribution} 

Building on our previous work in  \cite{rao2024exact}, we present a numerical approximation method to compute \eqref{eq_coro2}, using a Padé approximant for $\left<\exp(\lambda_{S};t,\theta)\right>$.

We start by defining the function: $$H(s,t)=\left(\nabla_{\boldsymbol{\theta}} +1\right)^{\textbf{Z}_0}\left<\exp(s\lambda_S(t;\boldsymbol{\theta}))\right>|_{\boldsymbol{\theta}=\textbf{0}}.$$ The required quantity for Eq. (\ref{eq_coro2}) is given by $H(1,t)$. To approximate $H(1,t)$, we first construct a Padé approximant for the function $H(s,t)$, denoted as $T(s,t)$, and then set $s=1$ to calculate $T(1,t)$. 

To obtain the Padé approximant of $H(s,t)$, we use the approach outlined in \cite{baker1981morris}. This approximant is based on the Maclaurin series expansion of $H(s,t)$ in $s$, truncated after $\tilde{N}$ terms: $T_{\tilde{N}}(s,t)= \sum_{n=0}^{\tilde{N}} \frac{s^n}{n!} \times \frac{\partial^n}{\partial s^n}H(s,t)|_{s=0} = \sum_{n=0}^{\tilde{N}}b_n(t)s^n$.

The Padé method approximates $H(s,t)$ by a rational function that shares the same Maclaurin series as $T_{\tilde{N}}(s,t)$. This rational function is expressed as $P^*_{\tilde{L}}(s,t)/Q^*_{\tilde{N}-\tilde{L}}(s,t)$, where $P^*_{\tilde{L}}(s,t)$ and $Q^*_{\tilde{N}-\tilde{L}}(s,t)$ are polynomials of degree $\tilde{L}$ and $\tilde{N}-\tilde{L}$, respectively, for some integer $\tilde{L}$.

The coefficients $b_n(t)$ of the Maclaurin series $T_{\tilde{N}}(s,t)$ are determined by:
$\left(\nabla_{\boldsymbol{\theta}} +1\right)^{\textbf{Z}_0}<\lambda_{S}(t)^n>$; since $$\frac{\partial^n}{\partial s^n}H(s,t)|_{s=0}= \left(\nabla_{\boldsymbol{\theta}} +1\right)^{\textbf{Z}_0}\left<\lambda_S(t;\boldsymbol{\theta}^n)\right>|_{\boldsymbol{\theta}=\textbf{0}}$$.

The key to constructing a Padé approximant for $H(s,t)$ lies in determining $<\lambda_{S}(t)^n>$. We follow the procedure from \cite{press1992flannery}, outlined as follows:

\begin{enumerate}
    \item   \textbf{Differentiate $\lambda_{S}(t)^n$:} Using Ito’s rule, we differentiate: ${\rm d}(\lambda_{S}^n)=n\lambda_{S}^{n-1}{\rm d}\lambda_{S}+ \frac{n(n-1)}{2}\lambda_{S}^{n-2}({\rm d}\lambda_{S})^2$, and substitute ${\rm d}\lambda_{S}$ and $({\rm d}\lambda_{S})^2$ from the SDEs in \eqref{sde_formu1}. The right-hand side becomes a polynomial in terms of $\lambda_{S}^{l_0}\boldsymbol{\lambda}^{\bm{l}}$.

\item \textbf{Recursive Differentiation:} Continue differentiating all new terms of the form $\lambda_{S}^{l_0}\boldsymbol{\lambda}^{\bm{l}}$ and substitute ${\rm d}\lambda_{S}$ and the SDEs from \eqref{sde_formu1} until all terms on the right-hand side are known.

\item \textbf{Averaging:} Take the average of the equations for ${\rm d}(\lambda_{S}^n)$ and ${\rm d}(\lambda_{S}^{l_0}\boldsymbol{\lambda}^{\bm{l}})$. This produces a set of ODEs whose solutions yield $<\lambda_{S}(t)^n>$ and $<\lambda_{S}^{l_0}\boldsymbol{\lambda}^{\bm{l}}>$, with the ODE system having finite dimension due to the linearity of the SDEs in \eqref{sde_formu1}. The initial values for $<\lambda_{S}^{l_0}\boldsymbol{\lambda}^{\bm{l}}>$ are zero for $l_0 > 0$, and are $\bm{\theta}^{\bm{l}}$ when $l_0 = 0$.

\end{enumerate}

Finally, we differentiate both sides of the ODEs governing $<\lambda_{S}(t)^n>$ by $\left(\nabla_{\boldsymbol{\theta}} +1\right)^{\textbf{Z}_0}$, and set $\boldsymbol{\theta}=\textbf{0}$ to compute $\left(\nabla_{\boldsymbol{\theta}} +1\right)^{\textbf{Z}_0}<\lambda_{S}(t)^n>|_{\boldsymbol{\theta}=\textbf{0}}$. This process leaves the format of the ODEs unchanged but modifies the initial conditions from $\bm{\theta}^{\bm{l}}$ to $\left(\nabla_{\boldsymbol{\theta}} +1\right)^{\textbf{Z}_0}\boldsymbol{\theta}^{\bm{l}}|_{\boldsymbol{\theta}=\textbf{0}}$.

\subsection{Analytical representation for the auxiliary chemical master equation with arbitrary initial conditions}

Although our previous numerical procedure in section B is only for delta initial distribution, in this section, we slightly modify the previous procedure to one for an arbitrary initial condition, which does not increase the computational complexity.

An arbitrary initial distribution is represented as a sum of delta distributions
\begin{equation}\label{eq_arbini}
    {\rm P}(\bm{x},t)=\sum_{\bm{y}}d_{\bm{y}}\delta_{\bm{y}}(\bm{x}),
\end{equation}
where the summation is over all possible system states, $\bm{y}$.

The analytical FPT distribution with the arbitrary initial distribution, \eqref{eq_arbini}, follows from \eqref{eq_coro2} and linearity of the dependence of the FPT distribution on the initial distribution:
\begin{equation}
\begin{split}\label{eq_coro3}
P(FPT>t)&=\sum_{\bm{y}}d_{\bm{y}}\left(\nabla_{\boldsymbol{\theta}} +1\right)^{\bm{y}}\left<\exp(\lambda_S(t;\boldsymbol{\theta}))\right>|_{\boldsymbol{\theta}=\textbf{0}}.
\end{split}
\end{equation}

The numerical procedure modifies the initial value of $<\lambda_{S}^{l_0}\boldsymbol{\lambda}^{\bm{l}}>$ in the previous procedure from $\left(\nabla_{\boldsymbol{\theta}} +1\right)^{\textbf{Z}_0}\boldsymbol{\theta}^{\bm{l}}|_{\boldsymbol{\theta}=\textbf{0}}$ to $\sum_{\bm{y}}d_{\bm{y}}\left(\nabla_{\boldsymbol{\theta}} +1\right)^{\bm{y}}\boldsymbol{\theta}^{\bm{l}}|_{\boldsymbol{\theta}=\textbf{0}}$

\subsection{Theory validation and simulation results}

We first applied our theory to an example of a biochemical network consisting of four zero/first-order reactions upstream of a second-order reaction:
\begin{equation}\label{exam_sys}
\begin{array}{lll}
\emptyset  &{\xtofrom[\ \;{ a_2(t)}\ ]{\ \;{ a_1(t)}\ }} & S_1 \vspace{1ex},\\
% S_1 &{\xrightarrow{\ \ \;{ a_2(t)}\ \ }} &\emptyset, \vspace{1ex}\\
\emptyset  &{\xtofrom[\ \;{ a_4(t)}\ ]{\ \;{ a_3(t)}\ }} & S_2 \vspace{1ex},\\
% S_2 &{\xrightarrow{\ \ \;{ a_4(t)}\ \ }} &\emptyset, \vspace{1ex}\\
S_1 + S_2&{\xrightarrow{\ \ \ \;{ a_0(t)}\ \ \ }} &S_0,
\end{array}
\end{equation}

We will determine the FPT distribution using Eq. (\ref{eq_coro1}) and Eq. (\ref{eq_coro2}), and we will compare these results with those obtained from simulations conducted using the Gillespie algorithm \cite{gillespie2013perspective, cao2006efficient}, which serves as our stochastic simulation algorithm (SSA).

We first investigated the FPT distributions of the system described in Eq. \ref{exam_sys} when the reaction rates are constants. We compared the FPT distributions resulting from delta initial distributions and those from Poisson-product form distributions, ensuring that all other factors remained constant. To facilitate a fair comparison, each component of the delta distribution was selected to be the mean of the corresponding marginal Poisson distribution used in the Poisson-product form.

\begin{figure}[htp]
\includegraphics[width=0.95\linewidth]{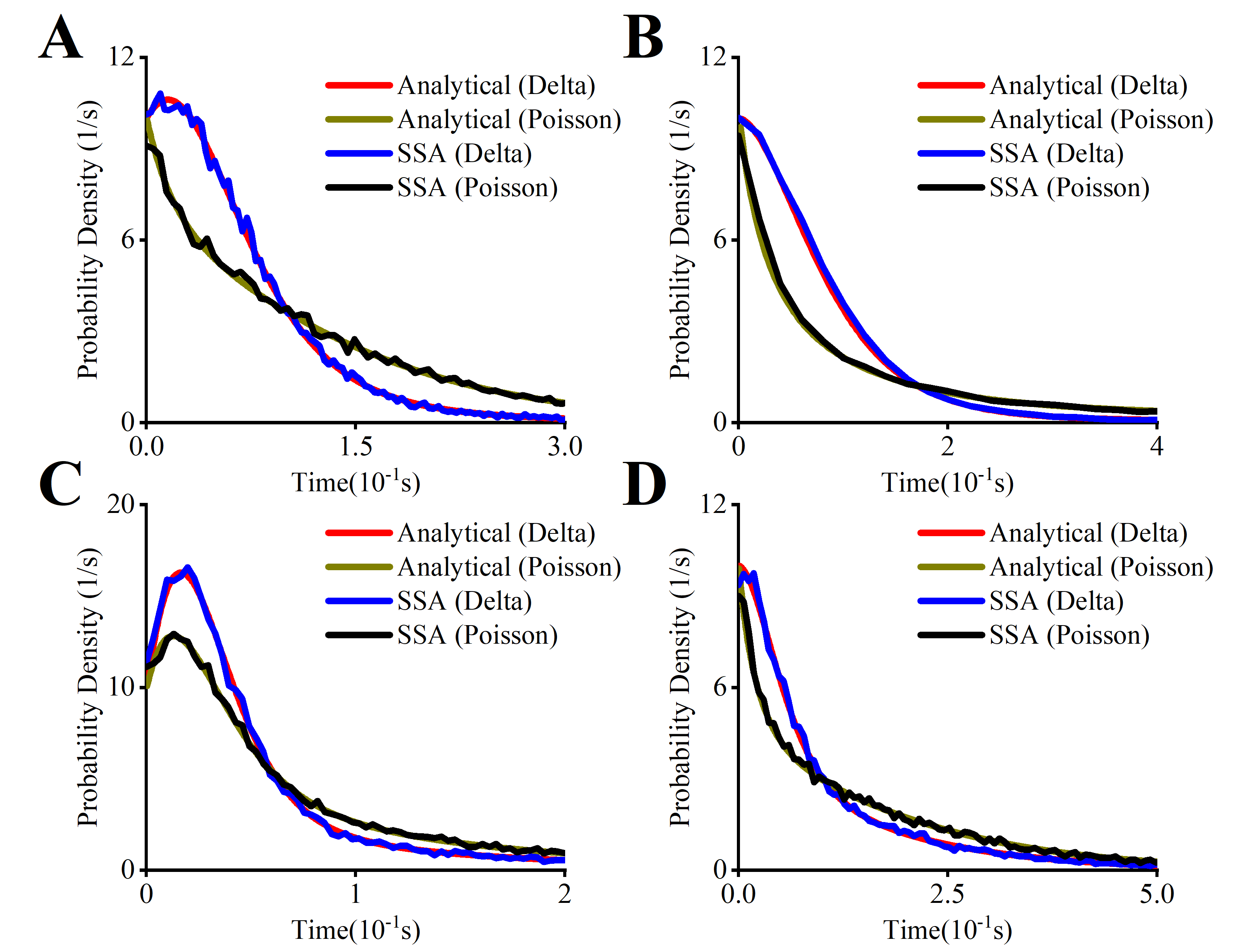}
\caption{\label{fig:delta-toy} First Passage Time (FPT) distributions for the example biochemical network described in Eq. \ref{exam_sys}, comparing cases where initial conditions are set to fixed values (delta distribution) versus a Poisson-product form. FPT distributions were calculated for four different parameter sets (A-D). In each parameter set, the fixed initial state value was chosen as the mean of the corresponding marginal Poisson distribution. The results show that variations in initial conditions lead to significant differences in the FPT distributions (blue lines vs. black lines). This further highlights the necessity of pushing the theoretical analysis to arbitrary initial conditions.} 
\end{figure}

First and foremost, our theoretical results align well with the stochastic simulations. Under the same controlled initial distributions, both Eq. (\ref{eq_coro1}) and Eq. (\ref{eq_coro2}) provide a good approximations of the corresponding FPT distributions obtained from the SSA simulations (as seen in Fig. \ref{fig:delta-toy}, where red lines correspond to the theoretical results and blue lines correspond to the simulated results, with gray and black lines representing additional comparisons).

The FPT distributions show significant differences across all four tested parameter sets when initial conditions change from one distribution to another (illustrated by the blue vs. black lines in Fig. \ref{fig:delta-toy}). This pattern appears to be robust, as the parameter sets were chosen randomly from diverse regions of the parameter space. Moreover, the FPT distribution statistics can vary substantially depending on the parameter selection. For example, in Fig. \ref{fig:delta-toy} A and B, both the mean and standard deviation can change by over 50\% (Tab. \ref{tab1}). Notably, in Fig. \ref{fig:delta-toy}, the mean FPT shifted by 186\%, while the standard deviation altered by 108\% when the initial condition changed from a Poisson product form to a delta distribution.

These dramatic statistical shifts underscore the significant impact of initial conditions on FPT distribution outcomes, highlighting the importance of our theoretical corrections for FPT distributions under different initial conditions. The pronounced sensitivity of FPT statistics to initial state changes further illustrates that full FPT distributions provide a more comprehensive reflection of system dynamics compared to relying solely on summary statistics, which are more susceptible to variations in the system.

\begin{table}[htp]
    \centering
    \caption{Statistics of the FPT distributions in Fig. \ref{fig:delta-toy} }
    \begin{threeparttable}
    \begin{tabular}{c|c|cccc}
    \toprule
    Statistic    & \makecell[c]{Initial or\\ Difference} & panel A &panel B  & panel C   & panel D\\
    \hline
    \multirow{3}{*}{ Mean}    &Poisson &0.11&0.33&-%0.069
                                &-%0.14
                                \\
                              & Delta  &0.064&0.11&-&-\\%0.053&0.11\\
                               & Difference &\textbf{69\%}        &\textbf{186\%}&
                    -&-\\
                    %\textbf{30\%}&      \textbf{36\%}\\  
                    %\cline{2-7}
                    \hline
                     \multirow{3}{*}{\makecell{Standard \\Deviation}} & Poisson  &0.097&0.45&-&-\\%0.070&0.13\\
                               & Delta  & 0.055&0.22&-&-\\%0.058&0.11\\
                               & Difference &\textbf{77\%}&        \textbf{108\%}&-&-\\
                    %\textbf{20\%}&      \textbf{18\%}\\                   
                    \hline
    \end{tabular}
    \begin{tablenotes}
    \footnotesize
     \item   Note: any changes below 50\% are not shown and are indicated with a dashed line ('-').
    \end{tablenotes}
    \end{threeparttable}
    \label{tab1}
\end{table}

Secondly, we investigated the FPT distributions of the system described in Eq. \ref{exam_sys} under conditions of time-varying reaction rates (see Fig. \ref{fig:time_var}). To simplify our analysis, we focused on varying the parameter $a_0$, while keeping all other rate parameters constant. We tested four distinct cases with different temporal profiles for $a_0$: ramping signals with two different rising speeds (Figs. \ref{fig:time_var}A \& B), a staircase signal (Fig.\ref{fig:time_var}C), and a sinusoidal signal (Fig. \ref{fig:time_var}D). Our theoretical results were validated against stochastic simulations in all tested cases (red vs. blue and gray vs. black lines in Fig. \ref{fig:time_var}), showcasing the robustness of our theoretical analysis. To our knowledge, this represents the first comprehensive theoretical analysis providing complete FPT distributions for time-varying biochemical systems with second-order reactions.

It is well established that reaction rates influence FPT distributions, but the interaction between reaction rates and initial conditions in shaping these distributions is complex. Notably, the FPT distributions vary with changes in the system's initial distributions, and the extent of these changes is dependent on the temporal patterns of $a_0$. Certain temporal profiles lead to more pronounced differences than others. For example, when $a_0$ follows a staircase pattern (Fig. \ref{fig:time_var}), the FPT distribution derived from Poisson-product form initial distributions exhibits significantly different skewness compared to that from delta initial conditions (Fig.\ref{fig:time_var}C). Interestingly, even without altering the temporal pattern of $a_0$, slight changes in its rising speed can yield different effects as a result of the initial condition distributions. The two FPT distributions resulting from the distinct initial distributions show greater divergence in Fig. \ref{fig:time_var}A than in Fig. \ref{fig:time_var}B.

\begin{figure}[htp]
\includegraphics[width=0.95\linewidth]{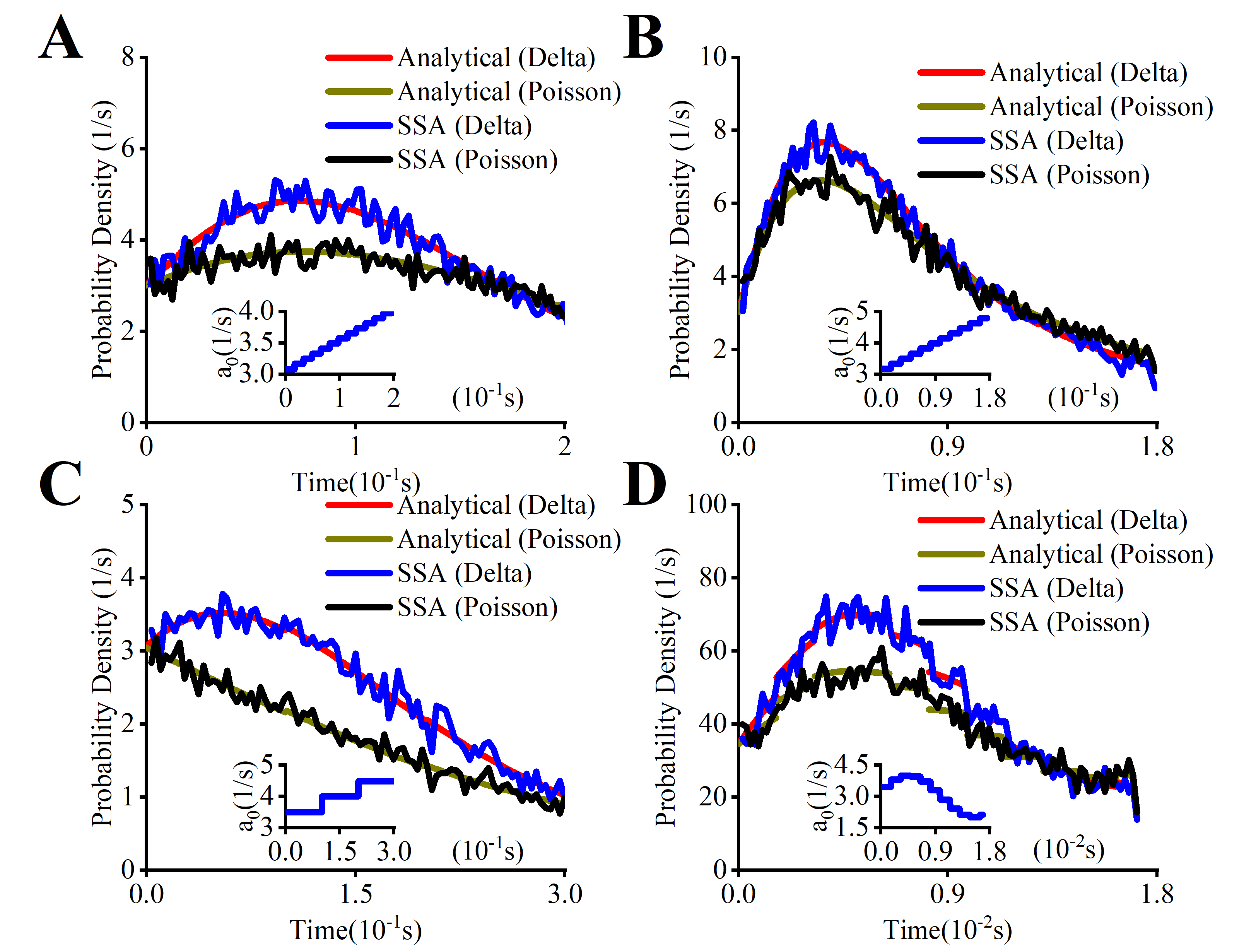}
\caption{\label{fig:time_var} First Passage Time (FPT) distributions for a biochemical network described by Eq. \ref{exam_sys} under conditions of time-varying reaction rates. Two types of initial conditions were evaluated: a fixed value (delta distribution) and a Poisson-product form. To simplify the analysis, only the time profile of $a_0$ —the rate parameter of the network’s second-order reaction—was varied. Four temporal profiles for $a_0$ were tested, including ramping signals at two different speeds (A-B), a staircase signal (C), and a sinusoidal signal (D). For each case, the fixed initial state corresponded to the mean of the associated Poisson marginal distribution. Our results demonstrate that the theoretical framework accurately predicts FPT distributions in time-dependent systems. Similar to cases with constant rates, variations in initial conditions lead to substantial differences in FPT distributions (blue vs. black lines), highlighting the robustness of the theoretical approach for both arbitrary initial conditions and time-dependent dynamics.} 
\end{figure}

\begin{figure}[htp]
\includegraphics[width=0.95\linewidth]{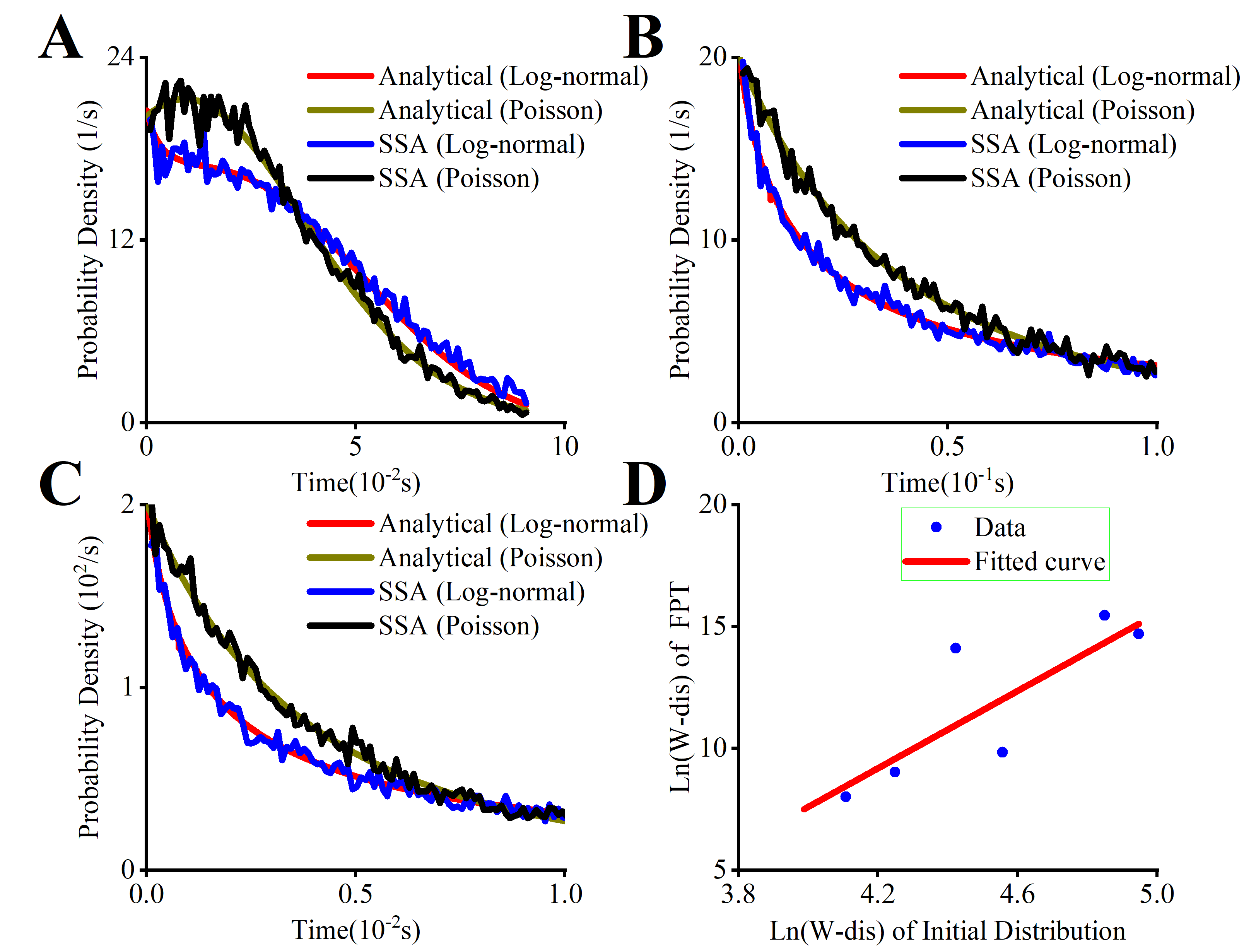}
\caption{\label{fig:logn-toy} First Passage Time (FPT) distributions for the biochemical network in Eq. \ref{exam_sys}, comparing initial conditions set as a lognormal-product form versus a Poisson-product form. FPT distributions were calculated for four parameter sets (A-D), with each state’s mean in the marginal lognormal and Poisson distributions set to be the same. The results indicate that switching from Poisson to lognormal initial conditions causes noticeable changes in the FPT distributions (blue lines vs. black lines), though these changes are less pronounced than those caused by delta-distributed initial conditions. Further analysis reveals that larger differences in initial conditions lead to greater differences in FPT distributions, indicating a consistent relationship between initial condition variations and FPT distribution changes. This monotonic relationship underscores the linear mapping between FPT distributions and initial conditions.} 
\end{figure}

These findings have important implications for the study of FPT distributions in time-varying biochemical systems. Relying solely on first- or second-order moments of FPT distributions is insufficient for understanding system dynamics and may fall short in comparative analyses. This inadequacy arises from the intricate relationships between these statistical measures, initial conditions, and time-varying rate parameters. Thus, there is a clear necessity to calculate the full FPT distributions to effectively capture the underlying dynamics.

Since our analysis can be generalized to accommodate any initial conditions, we simulated and analyzed the system described in Eq. \ref{exam_sys} with each state initialized using lognormally distributed variables (see Fig. \ref{fig:logn-toy}). We compared the FPT distributions resulting from lognormal initial states to those from Poisson-distributed initial states, ensuring that the means of each state variable remained consistent across different distributions. For simplicity, we tested the system with four sets of constant reaction rates. The varying initial distributions in all tested cases produced distinct FPT distributions. 

However, the statistical changes resulting from lognormal initial conditions (Table \ref{tab2}) are less pronounced compared to those from delta initial conditions (Table \ref{tab1}). For instance, in Fig. \ref{fig:logn-toy}A, the mean shift is only 11\%, whereas Fig. \ref{fig:delta-toy}A shows a significant 69\% shift. 

To investigate why certain initial distributions lead to more significant differences in FPT distributions than others, we explored whether there is an underlying principle governing this effect. We calculated the Wasserstein (W) distances between the initial conditions and the resulting FPT distributions. Our findings show that greater differences in initial conditions correlate with larger differences in the FPT distributions. This result aligns with the established understanding of a linear mapping between the system's initial conditions and its FPT distributions.

% \begin{figure}[htp]
% \includegraphics[width=0.95\linewidth]{logn_toy_df.png}
% \caption{\label{fig:logn-toy} First Passage Time (FPT) distributions for the biochemical network in Eq. \ref{exam_sys}, comparing initial conditions set as a lognormal-product form versus a Poisson-product form. FPT distributions were calculated for four parameter sets (A-D), with each state’s mean in the marginal lognormal and Poisson distributions set to be the same. The results indicate that switching from Poisson to lognormal initial conditions causes noticeable changes in the FPT distributions (blue lines vs. black lines), though these changes are less pronounced than those caused by delta-distributed initial conditions. Further analysis reveals that larger differences in initial conditions lead to greater differences in FPT distributions, indicating a consistent relationship between initial condition variations and FPT distribution changes. This monotonic relationship underscores the linear mapping between FPT distributions and initial conditions.} 
% \end{figure}

\begin{table}[htp]
    \centering
    \caption{Statistics of the FPT distributions in Fig. \ref{fig:logn-toy}}
    \begin{threeparttable}
    % \begin{tabular}{c|c|cccc}
    % \hline
    % Statistic    & Initial or Difference &\quad panel A \quad\  &\quad panel B \quad\  &\quad panel C  \quad\  &\quad panel D\quad\quad\\
    % \hline
    % \multirow{3}{*}{ Mean}    &Poisson& 0.028 &-&-&-\\    %0.035    &0.0035&   0.00038\\
    %                            & Log normal&0.031 &-&-&-\\   % 0.038&    0.0037&   0.00040\\
    %                            & Difference &\textbf{11\%}&    -&    -&    -\\
    %                 \hline
    %                  \multirow{3}{*}{\makecell{Standard \\Deviation}} & Poisson  &-&-&-&-\\% &0.020&0.028&0.0028&0.00032\\
    %                            & Log normal&-&-&-&-\\%0.021&0.030&0.0030&0.00034\\
    %                            & Difference &-&-&-&-\\                  
    %                 \hline
    % \end{tabular}
    \begin{tabular}{c|c|C{1.5cm}C{1.5cm}c}
    \toprule
    Statistic    & \makecell[c]{Initial or\\ Difference} &panel A &panel B  & panel C    \\
    \hline
    \multirow{3}{*}{ Mean}    &Poisson& 0.028 &-&-\\    %0.035    &0.0035&   0.00038\\
                               & Log normal&0.031 &-&-\\   % 0.038&    0.0037&   0.00040\\
                               & Difference &\textbf{11\%}&    -&    -\\
                    \hline
                     \multirow{3}{*}{\makecell{Standard \\Deviation}} & Poisson  &-&-&-\\% &0.020&0.028&0.0028&0.00032\\
                               & Log normal&-&-&-\\%0.021&0.030&0.0030&0.00034\\
                               & Difference &-&-&-\\                  
                    \hline
    \end{tabular}
     \begin{tablenotes}
    \footnotesize
     \item   Note: any changes below 10\% are not shown and are indicated with a dashed line ('-').
    \end{tablenotes}
    \end{threeparttable}
    \label{tab2}
\end{table}

Furthermore, we tested our theory on a simplified genetic regulation network model featuring a single gene that can exist in two states: inactive ($G$) and active ($G^*$)\cite{jia2024holimap}. In this model, a protein 
$P$, which is dynamically synthesized and degraded, can bind to $G$ to activate it. Once activated, the gene $G^*$ may revert to its inactive state after some time. The genetic regulation network (GRN) encompasses three chemical reactions, including a second-order reaction, as described in Eq. \eqref{grn}. Starting with substances that generate protein $P$, we aim to determine the timing of when protein $P$ successfully activates the gene.

\begin{equation}
\label{grn}
\begin{array}{lll}
\emptyset&\xtofrom[a2]{a1} &P\\
G^*&{\xrightarrow{\ \ a_3\ \ }} &G,\\
G +P&{\xrightarrow{\ \ a_0\ \ }} &G^*,\\
\end{array}
\end{equation}

We primarily investigated the impact of changes in initial distributions on the FPT distributions. Our findings indicate that all the conclusions drawn from testing the model in Eq. \ref{exam_sys} also apply to the genetic regulation network model in Eq. \ref{grn}:

\begin{enumerate}
    \item Our theoretical framework accurately predicts FPT distributions for all tested initial distributions, including Poisson product form distributions, delta distributions, and lognormal distributions (Fig. \ref{fig:delta-grn} and Fig. \ref{fig:logn-grn}). 
    \item Variations in initial conditions can lead to significant changes in the FPT distributions and their statistics (Tab. \ref{tab3} and Tab. \ref{tab4}). 
    \item In general, greater distances between the initial conditions correspond to larger differences in the resulting FPT distributions.
\end{enumerate}

\begin{figure}[htp]
\includegraphics[width=0.95\linewidth]{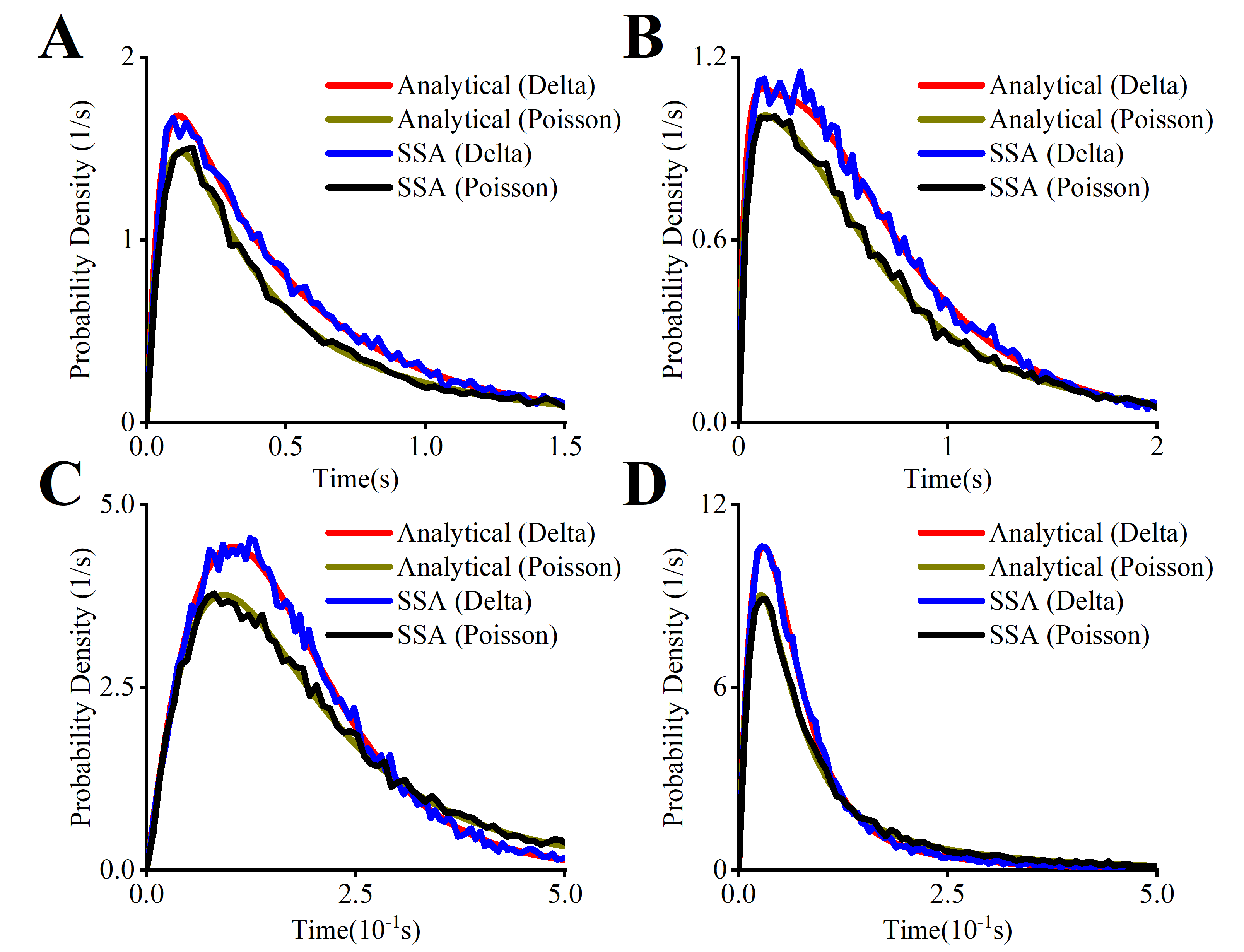}
\caption{\label{fig:delta-grn} First Passage Time (FPT) distributions for the genetic regulation network described in Eq. \ref{grn}, comparing cases where initial conditions are set to fixed values (delta distribution) versus a Poisson-product form. FPT distributions were computed for four different parameter sets (A-D). In each set, the fixed initial condition was chosen as the mean of the corresponding marginal Poisson distribution. The results demonstrate that differences in initial conditions lead to notable variations in FPT distributions (blue vs. black lines), underscoring the adaptability of our theoretical approach to various biochemical networks with arbitrary initial conditions.} 
\end{figure}

\begin{table}[htp]
    \centering
    \caption{Statistics of the FPT distributions in Fig.  \ref{fig:delta-grn}}
    \begin{threeparttable}
    \begin{tabular}{c|c|cccc}
    \toprule
    Statistic    & \makecell[c]{Initial or\\ Difference} & panel A & panel B  & panel C   & panel D\\
    \hline
    \multirow{3}{*}{ Mean}    &Poisson &-&-&-%0.57&      0.68&      0.20
    &      0.11\\
                              & Delta  &-&-&-%0.50&      0.62&       0.17
                              &     0.081\\
                               & Difference &-&-&-%\textbf{14\%}     &*\%&      \textbf{19\%}      
                               &\textbf{32\%}\\
                    \hline
                     \multirow{3}{*}{\makecell{Standard \\Deviation}} & Poisson  &-&0.58%&      0.62
                     &       0.15&      0.11\\
                               & Delta  &-& 0.42%&      0.48
                               &      0.10&     0.074\\
                               & Difference &-&\textbf{37\%} %     &\textbf{28\%}
                               &       \textbf{44\%}      &\textbf{50\%}\\                  
                    \hline
    \end{tabular}
    \begin{tablenotes}
    \footnotesize
     \item   Note: any changes below 30\% are not shown and are indicated with a dashed line ('-').
    \end{tablenotes}
    \end{threeparttable}
    % :below 30\%}
    \label{tab3}
\end{table}

\begin{figure}[htp]
\includegraphics[width=0.95\linewidth]{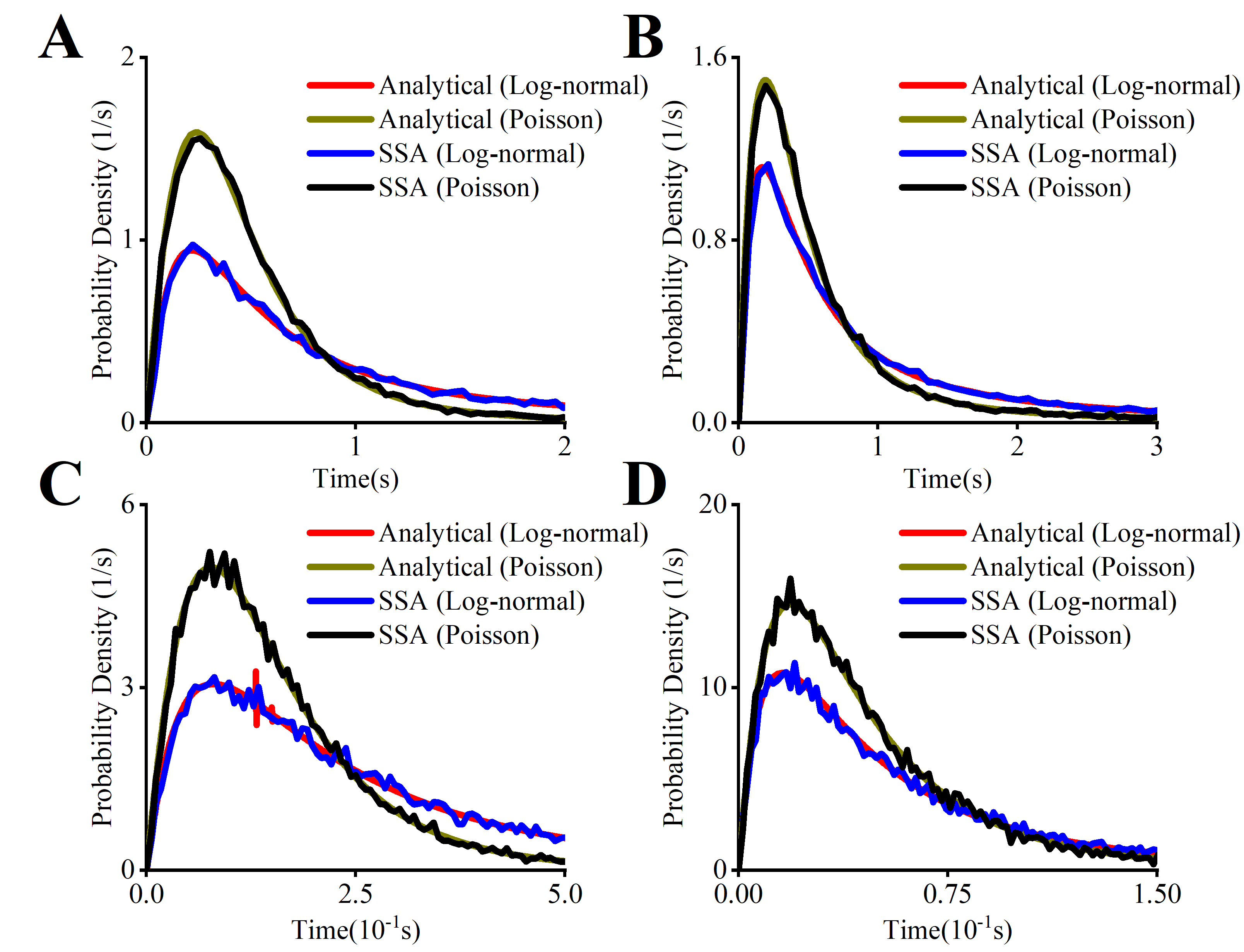}
\caption{\label{fig:logn-grn} First Passage Time (FPT) distributions for the genetic regulation network in Eq. \ref{grn}, comparing initial conditions set as a lognormal-product form versus a Poisson-product form. FPT distributions were calculated for four different parameter sets (A-D), ensuring that the mean of each state’s marginal lognormal and Poisson distributions was the same. The results show that switching from Poisson to lognormal initial conditions leads to noticeable changes in FPT distributions (blue vs. black lines). Unlike the previously studied biochemical network, some lognormal initial conditions resulted in more pronounced differences than delta initial conditions. Further analysis confirms that larger discrepancies in initial conditions correlate with greater differences in FPT distributions, supporting a consistent monotonic relationship between initial condition variation and changes in FPT distributions. This emphasizes the linear mapping between FPT distributions and initial conditions.} 
\end{figure}

\begin{table}[htp]
    \centering
    \caption{Statistics of the FPT distributions in Fig. \ref{fig:logn-grn}}
    \begin{threeparttable} 
    \begin{tabular}{c|c|cccc}
    \toprule
    Statistic    & \makecell[c]{Initial or\\ Difference} &panel A   & panel B   &\quad panel C    & panel D\\
    \hline
    \multirow{3}{*}{ Mean}    & Poisson &0.52&      0.64&    -&-\\%  0.15&     0.044\\
                               & Log normal&0.87&       1.0&   -&-\\%   0.21&     0.048\\
                               & Difference &\textbf{40\%}     &\textbf{39\%}     &-&-\\%\textbf{27\%}    &*\%\\
                    \hline
                     \multirow{3}{*}{\makecell{Standard \\Deviation}} & Poisson  &0.45&      0.70&   -&-\\%   0.10&      0.032\\
                               & Log normal&0.80&       1.2&    -&-\\%  0.14&     0.036\\
                               & Difference &\textbf{44\%}     &\textbf{42\%}     &-&-\\%\textbf{25\%}     &\textbf{11\%}\\                  
                    \hline
    \end{tabular}
    \begin{tablenotes}
    \footnotesize
     \item   Note: any changes below 30\% are not shown and are indicated with a dashed line ('-').
    \end{tablenotes}
    \end{threeparttable} 
    \label{tab4}
\end{table}

\section{Conclusion}

Initial conditions greatly shape the full FPT distributions in nonlinear biochemical networks. These conditions interact in a complex manner with nonlinear state transitions and potentially time-varying reaction rates to influence the time it takes for the system to reach a target state. This raises the question: how do initial conditions theoretically influence FPT distributions in nonlinear biochemical networks?

In this paper, we explored this question by deriving the full FPT distribution for a class of biochemical networks with arbitrary initial conditions. This work builds on our previous research, which provided an exact FPT distribution for nonlinear biochemical systems featuring an $A + B \rightarrow C$ type of second-order reaction downstream of a series of monomolecular reactions.

Our prior theoretical formulation, initially applied to Poisson-product form distributions, was adapted to include delta initial distributions, which can be expressed as a weighted sum of Poisson distributions. We then expanded the analysis to encompass arbitrary initial distributions, treated as combinations of discrete delta distributions.

Our findings demonstrate that even maintaining the same mean while varying the initial state distributions can result in significant shifts in the FPT distributions. Additionally, our theory describes how changes in initial conditions lead to variations in FPT distributions. We established that FPT distributions are linear with respect to initial conditions for the systems we studied, indicating that larger variations in initial conditions correspond to proportionally larger changes in FPT distributions.

Although our results are specific to systems conforming to \eqref{apb_reac}, to the best of our knowledge, this work is the first to derive exact FPT distributions for a general class of nonlinear biochemical networks with arbitrary initial conditions. Our results underscore the importance of theoretical studies on FPT distributions under varied initial conditions. We anticipate that future research will extend this type of theoretical analysis to a wider range of nonlinear biochemical networks, such as the commonly used enzymatic processes represented by the Michaelis-Menten reaction, where a first-order reaction follows a reversible second-order reaction.

\section*{Data Access and Conflict of Interest}
All relevant data are within the paper. The authors have declared that no competing interests exist. 

\section*{Acknowledgments}
\begin{acknowledgments}
Project supported by the Young Scientists Fund of the National Natural Science Foundation of China (Grant No. 12001111), Fund from the National Natural Science Foundation of China (Grant No. 11925103), Shanghai Municipal Science and Technology Major Project (No.2018SHZDZX01), ZJ Lab, and Shanghai Center for Brain Science and Brain-Inspired Technology, the 111 Project (No.B18015), and the 2021 STCSM (Grant No. 2021SHZDZX0103).
\end{acknowledgments}

\bibliography{apssamp}
\appendix
\begin{widetext}
\titleformat{\section}{\Large\bfseries}{}{0.3em}{\appendixname~\Alph{section}.}
\setcounter{proposition}{0}

\section{ Proof of Theorem \ref{theo2}}\label{Appendix:A}
We start from Proposition \ref{prop:delta_poisson}. 

\begin{proposition}\label{appprop:delta_poisson}
Any $N$-dimensional delta distribution, $\delta_{\bm{Z}_0}$, can be represented as a weighted sum of distributions of the Poisson-product form (refer to Appendix \ref{Appendix:B} for proof):

\begin{equation}\label{app:delta_poisson}
\begin{split}
\delta_{\bm{Z}_0}(\bm{x})
&=\sum_{k_1,k_2,\dots,k_N =0}^{\infty} c_{\bm{k}}\prod_{j=1}^N \frac{(\theta^{(k_j)}_j)^{x_j}}{x_j!}\exp(-\theta^{(k_j)}_j ),
\end{split}
\end{equation}
where $\bm{k}=[k_1,\dots,k_N]^\top$ is the index vector; $c_{\bm{k}}$ is the weight coefficient; and $\theta^{(k_j)}_j$ is  {\color{black}the parameter of the $k\operatorname{th}$} distribution component within the Poisson-product-form.
\end{proposition}

Because of Proposition $1$, any high-dimensional delta initial conditions of system \eqref{apb_reac} can be represented as a weighted sum of distributions of the Poisson-product form. That is, with $\theta_j^{(k)}$ the parameter of the $k$th Poisson distribution, we have

% {\color{red} perhaps explain here that $\theta^k$ is not $\theta$ to power $k$. I suggest using $\theta^{(k)}$ 
% and explain significance of $\theta^k$}

\begin{equation}\label{appeq_ini_d2p}
\begin{split}
&P(\bm{x},0)=\prod_{j=1}^N\delta_{z_j}(x_j)=%\sum_{k_1,k_2,\dots,k_N =0}^{\infty}
\sum_{\mathbf{k}}c_{\bm{k}} \prod_{j=1}^N \frac{(\theta^{(k_j)}_j)^{x_j}}{x_j!}\exp(-\theta^ {(k_j)}_j ).
\end{split}
\end{equation}

If the initial conditions, of the system in \eqref{apb_reac}, follow a distribution with a format of the $k$th term in \eqref{appeq_ini_d2p}, i.e. Poisson-product form, we have a FPT distribution: 

\begin{equation}\label{appeq:fpt_deexp}
\begin{split}
{\rm P}\left(FPT_{\bm{k}}>t\right)=\left<\exp\left(\lambda_S(t;\bm{\theta}^{(\mathbf{k})})\right)\right>.
\end{split}
\end{equation}

The linear mapping between the FPT distribution and the system's initial conditions guarantees the superposition principle, leading to a result that the FPT distribution for an arbitrary delta initial condition ($\delta_{\mathbf{Z}_0}$) can be expressed as a sum of the FPT distributions that result from the Poisson-product form initial conditions (\ref{appeq:fpt_dec}). 

\begin{equation}\label{appeq:fpt_dec}
\begin{split}
{\rm P}\left(FPT>t\right)=\sum_{\mathbf{k}}c_{\bm{k}}{\rm P}\left(FPT_{\bm{k}}>t\right)=\sum_{\mathbf{k}}c_{\bm{k}}\left<\exp\left(\lambda_S(t;\bm{\theta}^{(\mathbf{k})})\right)\right>,
\end{split}
\end{equation}
where $\mathbf{k} =[k_1,k_2,\dots,k_N]$ is the index vector in \eqref{appeq_ini_d2p}, and $\bm{\theta}^{(\mathbf{k})}=[\theta_1^{(k_1)},\theta_2^{(k_2)},\dots,\theta_N^{(k_N)}]^\top$ is the mean parameter within the Poisson-product form distributions.  

Next, we prove that \eqref{appeq:fpt_dec} is equivalent to the result in Theorem 1. 

Firstly, $\lambda_{S}$ is governed by linear SDEs in \eqref{sde_formu1}, which can be written as a general form:

\begin{equation}\label{appeq:lambda_sde}
\begin{split}
{\rm d}\boldsymbol{\Lambda}(t)&=\mathcal{H}(t)\boldsymbol{\Lambda}(t)+\mathcal{I}(t),
\end{split}
\end{equation}

where $\boldsymbol{\Lambda}(t)= \left[%\begin{array}{c}
     \lambda_{S}(t)  ,
     \boldsymbol{\lambda}(t)^\top 
%\end{array} 
\right]^\top$, $\mathcal{H}(t)\boldsymbol{\Lambda}(t)$ is the homogeneous part and $\mathcal{I}(t)$ is the inhomogeneous part.

Then, from the Duhamel's principle, $\boldsymbol{\Lambda}(t)$ can be written as a linear function of its initial conditions:
\begin{equation}\label{appeq:lambda_duh}
\begin{split}
\boldsymbol{\Lambda}(t)=G(t,0)\boldsymbol{\Lambda}(0)+\int^t_0G(t,\tau)\mathcal{I}(\tau),
\end{split}
\end{equation}
where $G(t,\tau)\boldsymbol{\Lambda}(\tau)$ is the solution of the homogeneous SDE:
\begin{equation}\label{appeq:lambda_homo_sde}
\begin{split}
{\rm d}\boldsymbol{\Lambda}(t)&=\mathcal{H}(t)\boldsymbol{\Lambda}(t),
\end{split}
\end{equation}
with the initial condition at time $\tau$: $\boldsymbol{\Lambda}(\tau)$.

Because $\boldsymbol{\Lambda}(t)= \left[
%\begin{array}{c}
     \lambda_{S}(t)  ,
     \boldsymbol{\lambda}(t)^\top 
%\end{array} 
\right]^\top$, and $\boldsymbol{\Lambda}(0)= \left[
%\begin{array}{c}
     0  ,
    \bm{\theta}^\top
%\end{array}
\right]^\top$, $\lambda_{S}(t;\bm{\theta})$ can be written as a linear function of $\bm{\theta}$:  

\begin{equation}\label{appeq:lambda_theta}
\begin{split}
\lambda_{S}(t;\bm{\theta})&=\mathcal{W}(t) \bm{\theta}  + \lambda_{S}(t;\bm{0}),
\end{split}
\end{equation}
where $\mathcal{W}(t)$ is row vector with n elements. 

Substituting  \eqref{appeq:lambda_theta} into \eqref{appeq:fpt_dec}, we have 
\begin{equation}\label{appeq:fpt_dec_simp1}
\begin{split}
{\rm P}\left(FPT>t\right)&=\left<\sum_{\mathbf{k}} c_{\bm{k}} \exp\left(\lambda_S(t;\bm{\theta}^{\mathbf{(k)}}\right)\right>\\& =\left<\exp\left(\lambda_{S}(t;0)\right)\sum_{\mathbf{k}} c_{\bm{k}}\prod_{j=1}^N\exp((\mathcal{W}_{j}(t)+1)\theta_j^{(\mathbf{k})})\exp(-\theta_j^{(\mathbf{k})})\right>\\
&=\left<\exp\left(\lambda_{S}(t;0)\right)\sum_{\mathbf{k}} c_{\bm{k}}\prod_{j=1}^N\sum_{x_j =0}^{\infty}\frac{(\mathcal{W}_{j}(t)+1)^{x_j}(\theta_j^{(\mathbf{k})})^{x_j}}{x_j!}
\exp(-\theta_j^{(\mathbf{k})})\right>\\
&=\left<\exp\left(\lambda_{S}(t;0)\right)\sum_{x_1,x_2,\dots,x_N =0}^{\infty}\prod_{j=1}^N(\mathcal{W}_{j}(t)+1)^{x_j}\sum_{\mathbf{k}} c_{\bm{k}}\frac{(\theta_j^{(\mathbf{k})})^{x_j}}{x_j!}\exp(-\theta_j^{(\mathbf{k})})\right>\\
&=\left<\exp\left(\lambda_{S}(t;0)\right)\sum_{x_1,x_2,\dots,x_N =0}^{\infty}(\mathcal{W}(t)+1)^{\mathbf{x}}\delta_{\mathbf{Z}_0}(\mathbf{x})\right>\\
&=\left<\exp\left(\lambda_{S}(t;0)\right)(\mathcal{W}(t)+1)^{\mathbf{Z}_0}\right>\\
\end{split}
\end{equation}
where $\mathcal{W}_{j}(t)$ is the $j\operatorname{th}$ component of $\mathcal{W}(t)$.

Then theorem 1 gets proved, as \eqref{appeq:fpt_dec_simp1} is derived as
\begin{equation}%\label{appeq:fpt_dec}
\begin{split}
{\rm P}\left(FPT>t\right)&=\left<(\mathcal{W}(t)+1)^{\mathbf{Z}_0}\exp\left(\lambda_{S}(t;\bm{0})+\mathcal{W}(t)\bm{\theta}\right)\right>|_{\bm{\theta}=\bm{0}}\\&=(\nabla_{\bm{\theta}}+1)^{\mathbf{Z}_0}\left<\left.\exp\left(\lambda_S(t;\bm{\theta})\right)\right|_{\bm{\theta}=\bm{0}}\right>
\end{split}
\end{equation}

\section{ Proof of Proposition \ref{prop:delta_poisson}}\label{Appendix:B}

In this section, we prove Proposition \ref{prop:delta_poisson}, which states that a high-dimensional delta distribution can be expressed as a weighted sum of distributions of Poisson-product form.

The proof is based on the P\'olya's theorem (Lemma \ref{applemma:polya}). 

\begin{lemma}(P\'olya's theorem :Theorem 5.3.1 in \cite{davis1963})\label{applemma:polya}
Let there be given an infinite set of linear equations in infinitely many unknowns $x_1,x_2,\dots$:
\begin{equation}\label{appeq:lineareq}
    \begin{array}{ccc}
        a_{11}x_1&+a_{12}x_2+\dots=&b_1,\\
        a_{21}x_1&+a_{22}x_2+\dots=&b_2,\\
        \vdots&&\vdots
    \end{array}
\end{equation}
No assumptions are made about $b$'s, but as far as $a$'s are concerned we assume
    
    \textbf{(A)} Let $q\geq 0$ and $n\geq 0$ be arbitrary integers. From the infinite array of coefficients
    \begin{equation}
    \begin{split}
        &a_{1,q+1},a_{1,q+2},\dots\\
        &\vdots\\
        &a_{n,q+1},a_{n,q+2},\dots\\
    \end{split}
\end{equation}
We may select $n$ columns such that the determinant formed by these columns does not vanish.

\textbf{(B)} For $i=2,3,\dots$, we have
\begin{equation}\label{appeq:assumpB}
    \begin{split}
    \lim_{j\rightarrow\infty}\frac{a_{i-1,j}}{a_{ij}}=0
    \end{split}
\end{equation}
Under assumptions (A) and (B), we may find a solution $x_i$ to \eqref{appeq:lineareq} with all the infinite series absolutely convergent.
\end{lemma}

\begin{proof} Because a multidimensional delta distribution is a product of each component's one-dimensional delta distributions (\ref{appeq_b_d2p}):

\begin{equation}\label{appeq_b_d2p}
\begin{split}
&\delta_{\bm{Z}_0}=\prod_{j=1}^N\delta_{z_j}(x_j)=\sum_{k_1,k_2,\dots,k_N =0}^{\infty} c_{\bm{k}}\prod_{j=1}^N \frac{(\theta^{(k_j)}_j)^{x_j}}{x_j!}\exp(-\theta^{(k_j)}_j ),
\end{split}
\end{equation}
where $c_{\bm{k}}=\prod_{j=1}^N c_{j,k_j}$, we only need to prove that a delta distribution is a weighted sum of Poisson distributions: 
\begin{equation}\label{appeq:delta_poisson}
\begin{split}
\delta_{z}(x)
&=\sum_{k =0}^{\infty} c_{k} \frac{(\theta^{(k)})^{x}}{x!}\exp(-\theta^{(k)} ).
\end{split}
\end{equation}

To prove \eqref{appeq:delta_poisson}, we need to find a series of coefficients $c_k$ for $k=0,1,\dots$ that satisfies \eqref{appeq:delta_poisson}. 

Let $\theta^{(k)}=k+1$, then $y_k=c_k\exp(-(k+1))$. Finding coefficient series $c_k$ is then equivalent to proving the existence of unknown series for $k=0,1,\dots$ from the infinite-dimensional linear equations:
\begin{equation}\label{appeq:lineareq1}
    \begin{array}{ccc}
        a_{11}y_0&+a_{12}y_1+\dots=&\delta_z(0),\\
        a_{21}y_0&+a_{22}y_1+\dots=&\delta_z(1),\\
        \vdots&&\vdots
    \end{array}
\end{equation}
    where $a_{i,j}$ for $i,j=1,2,\dots$ are chose as 
\begin{equation}\label{appeq:linear_coe_poiss}
a_{i,j}=\frac{(j)^{i-1}}{(i-1)!}.
\end{equation}

From Lemma \ref{applemma:polya}, {$y_k$} exists as long as \eqref{appeq:lineareq1} satisfies assumptions (A) and (B).

Assumption (A) requires that the matrix $\mathcal{A}$, which is composed of $\{a_{ij}\}$ $(i=1,\dots,n, j=q,q+1,\dots,q+n-1)$, is in full rank for $n>0$ and $q>0$  arbitrary integers. 

\begin{equation}
\mathcal{A}=\left(
\begin{array}{cccc}
    1 & 1,&\dots&1  \\
     q & q+1,&\dots&q+n-1  \\
     \vdots&\vdots&\ddots&\vdots\\
     \frac{(q)^{n-1}}{(n-1)!} & \frac{(q+1)^{n-1}}{(n-1)!},&\dots&\frac{(q+n-1)^{n-1}}{(n-1)!}  \\
\end{array}
\right)=\mathcal{D}\mathcal{M},
\end{equation}
where $\mathcal{D}$ is the diagonal matrix: ${\rm diag}\{1,1,\frac{1}{2},\dots,\frac{1}{(n-1)!} \}$, and $\mathcal{M}$ is a Vandermonde matrix: 

\begin{equation}
\mathcal{M}=\left(
\begin{array}{cccc}
     1 & 1&\dots&1  \\
     q & q+1&\dots&q+n-1  \\
     \vdots&\vdots&\ddots&\vdots\\
     (q)^{n-1} & (q+1)^{n-1}&\dots&(q+n-1)^{n-1}  \\
\end{array}
\right).
\end{equation}

Because both $\mathcal{D}$ and $\mathcal{M}$ are of full rank, $\mathcal{A}$ is of full rank, satisfying assumption (A).

Assumption (B) follows naturally when we substitute \eqref{appeq:linear_coe_poiss} into \eqref{appeq:assumpB}: 
\begin{equation}
    \begin{split}
    \lim_{j\rightarrow\infty}\frac{a_{i-1,j}}{a_{ij}}=\lim_{j\rightarrow\infty}\frac{(i-1)}{j}=0.
    \end{split}
\end{equation}
\end{proof}
\end{widetext}
% Produces the bibliography via BibTeX.

\end{document}